\documentclass[12pt,fleqn]{article}

\usepackage{amssymb,amsmath}

\usepackage{natbib}
\usepackage{vmargin,setspace}

\usepackage{libertine}

\usepackage[colorlinks=true, pdfstartview=FitV, citecolor=blue!50!black, linkcolor=red!30!black]{hyperref}

\usepackage{graphicx,graphics,tikz}

\newcounter{llst}

\newenvironment{numm}{\begin{list}{\rm (\roman{llst})}{\usecounter{llst}
\setlength{\itemindent}{0em} \setlength{\leftmargin}{3.5em}
\setlength{\labelwidth}{2.5em} \setlength{\labelsep}{1em}}}{\end{list}}

\newtheorem{theorem}{Theorem}[section]

\newtheorem{axiom}[theorem]{Axiom}

\newtheorem{corollary}[theorem]{Corollary}

\newtheorem{definition}[theorem]{Definition}
\newtheorem{expl}[theorem]{Example}

\newtheorem{rmrk}[theorem]{Remark}

\newenvironment{proof}[1][Proof]{\noindent \textbf{#1.} }{\hfill
\rule{0.5em}{0.5em}}
{}
\newenvironment{remark}{\begin{rmrk} \rm}{\hfill $\blacklozenge$
\end{rmrk}}{}

\setpapersize{A4}
\setmarginsrb{84pt}{64pt}{84pt}{72pt}{12pt}{10pt}{12pt}{30pt}

\begin{document}

\title{\textbf{Platform Competition as \\ Network Contestability}\thanks{We thank Matthias Beck and Owen Sims for valuable discussions on the subject of this paper and Lydia Pik Yi Cheung for literature suggestions. The comments of seminar participants at Temple University and the University of Birmingham have also been useful. We also thank Marianne Miserandino for her helpful suggestions on the presentation.}}

\author{Robert P.~Gilles\thanks{Queen's University Management School, Riddel Hall, 185 Stranmillis Road, Belfast BT9\ 5EE, UK. Email: \textsf{r.gilles@qub.ac.uk}} \and Dimitrios Diamantaras\thanks{Department of Economics, Temple University, Philadelphia, PA. Email: \textsf{dimitris@temple.edu}} }

\date{October 2013}

\maketitle

\begin{abstract}
\singlespace\noindent
Recent research in industrial organisation has investigated the essential place that middlemen have in the networks that make up our global economy. In this paper we attempt to understand how such middlemen compete with each other through a game theoretic analysis using novel techniques from decision-making under ambiguity.

We model a purposely abstract and reduced model of one middleman who provides a two-sided platform, mediating surplus-creating interactions between two users. The middleman evaluates uncertain outcomes under positional ambiguity, taking into account the possibility of the emergence of an alternative middleman offering intermediary services to the two users.

Surprisingly, we find many situations in which the middleman will purposely extract maximal gains from her position. Only if there is relatively low probability of devastating loss of business under competition, the middleman will adopt a more competitive attitude and extract less from her position.

\end{abstract}

\newpage

\section{Platform provision and contestation}

\noindent
Middlemen arise in many economic and social situations. Usually the emergence of middlemen in economic situations relates to the nature of supply chains in our networked, global economy. An individual who occupies a middleman position in a trade network can interrupt trade as well as information flows and, therefore, has power due to her position in this trade network. This power implies that a middleman can command a larger share of the gains from trade under her control. By withdrawing from a trade chain and disabling the platform she provides, the middleman can disrupt the interaction between the users under consideration, which gives the middleman high bargaining power in comparison with these users.

The position a middleman occupies in the network explains her ability to levy high access and trading fees in the market. This explains the high profits of NYSE and NASDAQ as well as the success of e-Bay and Amazon.com. It also explains why Google and Facebook are viewed as potentially extremely profitable; they control vast amounts of information about individuals who participate in their online environments.

Recent research in systems science, computer science and physics has investigated the structure of socio-economic networks that make up the global economy. Evidence for the importance of middleman positions in these networks has been found and used to introduce theories that try to understand the functionality of these networks from this perspective \citep{Newman2006book}.\footnote{We refer to \citet{Newman2010book} for an accessible overview of these theories and perspectives.} For example, \citet{Vitali2011} and \citet{Vitali2013} investigated the corporate ownership network in the global economy and identified the central position of middlemen in a so-called strongly connected core. They found that financial services multinationals occupy positions in this strongly connected core.  Clearly, this refers to the control that these banks as middlemen exert in these ownership networks.

\subsection{Middlemen as platform providers}

\noindent
The main reason that middlemen emerge in the socio-economic networks surrounding us, is that they facilitate mutually beneficial interactions that make up our global economy. Although the nature of their privileged positions is well understood, it is much less clear how middlemen compete or are contested. In this paper, we propose that a middleman's perception of potential contestation of his privileged position filters the way competition limits the power of the middleman. We base our argument on the principle that contestability is all in the mind of the middleman and this may make the middleman slow to react to emerging competition for the middleman position. We make essential use of tools from the theory of decision making under ambiguity.

Our point of departure is that middlemen can be fully characterised as providers of \emph{interaction platforms}. As such our approach complements the theory of trade platforms seminally introduced in \citet{Tirole2003}, developed further for two-sided platforms in \citet{EvansSchmalensee2007}, \citet{Rysman2009} and \citet{Damme2012} and for multi-sided platforms in \citet{HagiuWright2011} and \citet{EvansSchmalensee2013}. We restrict our discussion to consider a simplified model of a middleman as a provider of a two-sided interaction platform and, as such, a middleman is viewed as a facilitator of value-generating interaction between users of that platform. Our analysis of this intermediated interaction situation will be rather different from that employed in the quoted literature on two-sided platforms.

Two-sided platforms are usually considered to be market organisations that facilitate the confrontation of market demand and supply. Platform providers are assumed to charge usage fees for the market auctioneering services provided. Our approach fits best with the definition of a platform developed in \citet[page 2]{HagiuWright2011}, who view a platform as ``an organisation that creates value primarily by enabling direct interactions between two (or more) distinct types of affiliated customers''.

We consider the simple case of a single middleman $M$, who provides a platform for the mutually beneficial interaction of two users, denoted as $1$ and $2$. The network representation is depicted in Figure~\ref{fig:1}.

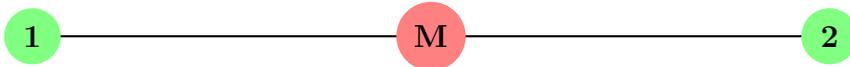
\begin{figure}[h]
\begin{center}
\begin{tikzpicture}[scale=0.35]
\draw[thick] (0,0) -- (15,0) -- (30,0);
\draw (0,0) node[circle,fill=green!50] {$\mathbf{1}$};
\draw (15,0) node[circle,fill=red!50] {$\mathbf{M}$};
\draw (30,0) node[circle,fill=green!50] {$\mathbf{2}$};
\end{tikzpicture}
\end{center}
\caption{Intermediated interaction between users $1$ and $2$.} \label{fig:1}
\end{figure}

We explicitly assume that the two users cannot interact directly, but only through the intermediation of the middleman.\footnote{An alternative interpretation would be that these users face insurmountable \emph{transaction costs} to overcome obstacles to their direct interaction.} This implies that the middleman has power of control by virtue of her network position and can exploit that control to her benefit. Such a middleman can be interpreted as a market-maker, although the platform does not necessarily have to be a marketplace. For instance, the landlord of a public house---or ``pub''---provides a platform for social interaction, which might or might not be economic in nature. The pub facilitates the creation of gains from interaction and, therefore, the landlord can expect a share in these gains through the sale of products in the pub.\footnote{This example of mainly social interaction clarifies that middleman intermediation usually involves investment to provide a platform. See \citet{GillesDiamantaras2003} for a formal model of the costly provision of such market-making in a non-network model.}

In the abstract social network depicted in Figure 1, we can show in a simple non-cooperative game that the middleman $M$ normally would exert fully her positional power and, consequently, would extract the full monopoly profits from the interaction---such as the only pub in a small town on the King's Road would do. The only limitation to this exertion of control is that the two users might have access to an alternative interaction platform, provided by another middleman. This form of contestation represents the most basic form of competition between multiple middlemen in a network.

One would expect that the access of users to multiple middlemen would be sufficient to guarantee a form of Bertrand-like price competition between intermediating agents, resulting in a reduction of access fees to marginal cost levels. However, observations from platform competition in our global economy indicate that there remain relatively high access fees in place. This was pointed out for payment card services by \citet{Tirole2003} and \citet{Rysman2009}. Rysman remarks that access pricing is affected by how users move between different platforms. In real world platforms, users move all \emph{or\/} some of their usage from one platform to the other. \citep[page 130]{Rysman2009}

\medskip\noindent
We develop a simple game theoretic model that uses some of these features. Our point of departure is that there are essentially two very different viewpoints to middleman competition with differing implications. First, from a static point of view, competition could refer to the absence of any middleman power in the interaction network. This implies that all interaction should either be two-sided---like a matching market---or based on chains with an abundance of multiple platform choices to all users. Second, from a more dynamic point of view, competition could refer to the principle that every middleman position should be challenged if it occurs. The second principle---akin to the classical notion of ``contestability'' \citep{Baumol1982}---is much harder to embed in a mathematical theory.

Our aim in this paper is to devise a game theoretic model that can encompass some basic features of the dynamic approach to platform competition. For this, imagine that a middleman is able to extract excess rents from the provision of a platform within the prevailing interaction network. However, if this position can be challenged, this is sufficient to force the middleman to reduce her transfer fees to a minimum. Thus, the middleman anticipates the emergence of a competing platform, which is prevented by the middleman through the reduction of the charged access fees to her platform. We denote this notion as that the middleman's position is \emph{contestable} in the prevailing network.

For example, Amazon.com is the main book seller on the Internet. So, in principle Amazon.com can be viewed as a middleman providing a platform in the global book market. However, Amazon's position is contestable, since if it charges higher prices, its potential buyers will resort to purchasing their books at vendors such as local bookstores or other online book sellers. So, even though many of its clients might be purchasing books exclusively from Amazon.com, a price increase might allow these economic agents to change their purchasing decisions and seek trade relations with alternative sellers even if they have no previous experience with these alternative sellers. It would also motivate others to start up new online bookstores.

The Amazon.com example refers to a change of behaviour and the activation of alternative, potential trade relations to circumvent a provided platform in the prevailing trade network. This behaviour is at the foundation of a dynamic notion of platform competition. This is depicted in Figure~\ref{fig:2}: The dashed relationships between the two users $1$ and $2$ with the alternative middleman $C$ need to be activated in order to undermine $M$'s middleman position.

\subsection{Our modelling approach}

\noindent
We develop a game theoretic approach to the modelling of middleman contestability in intermediated interaction situations. as stated, we limit our model to that of one middleman and two platform users. We describe the gains from interaction through a hedonic model. The two users are assumed to generate a certain value from interaction, which can be appropriated by the middleman through the imposition of a access fee to use the interaction platform the middleman provides. We allow extraction up to all of the gains the interaction generates and assume price discrimination. As expected, the standard Nash equilibrium in this model without alternative interaction platforms offered to the users, the middleman will extract all generated gains. This Nash equilibrium describes the benchmark case in our model and refers to the monopolistic case with a single middleman or interaction platform.

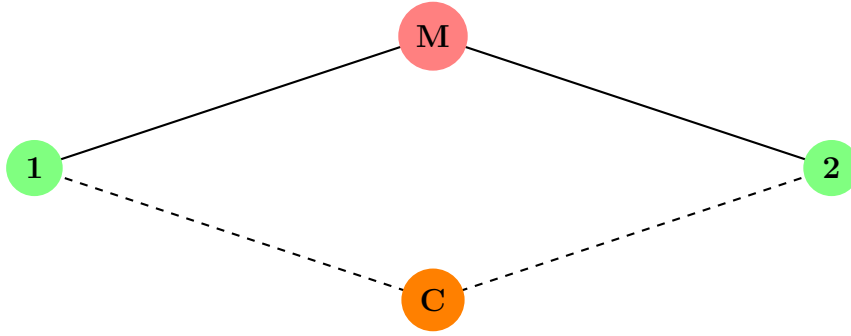
\begin{figure}[h]
\begin{center}
\begin{tikzpicture}[scale=0.35]
\draw[thick] (0,0) -- (15,5) -- (30,0);
\draw[thick,dashed] (0,0) -- (15,-5) -- (30,0);
\draw (0,0) node[circle,fill=green!50] {$\mathbf{1}$};
\draw (15,5) node[circle,fill=red!50] {$\mathbf{M}$};
\draw (15,-5) node[circle,fill=orange] {$\mathbf{C}$};
\draw (30,0) node[circle,fill=green!50] {$\mathbf{2}$};
\end{tikzpicture}
\end{center}
\caption{Contested interaction between users $1$ and $2$.} \label{fig:2}
\end{figure}

Next, we introduce a potential alternative middleman into the model, as depicted in Figure~\ref{fig:2}. We use the ambiguity approach introduced by \citet{EichbergerKelsey2002} to describe the belief system of a contested middleman. Hence, we assume that the middleman formulates optimistic and pessimistic beliefs regarding her position in the interaction network. The optimistic beliefs represent an uncontested middleman position with full monopolistic extraction depicted in Figure~\ref{fig:1}, while her pessimistic beliefs concern a contested middleman position depicted in Figure~\ref{fig:2}.

In this model, middleman contestation is introduced as a degree of pessimism---expression the expectation of the middleman to be in a contested position---as well as the expected level of interaction under contestation. The latter is understood as a level of \emph{customer loyalty} to the platform provided by the middleman. Our approach founded on this form of ambiguity as a form of boundedly rational decision making represents the idea that middleman contestability is ``all in the mind of the middleman''.

Our {\em main theorem\/} concludes that if both the degree of pessimism and the expected level of loyalty under contestation are sufficiently low, the middleman continues to maximally extract gains from her position. This corresponds to the case that the middleman takes full advantage of loyal customers, thereby lowering the quality of the service provided as well as increasing the charged fees. As a real-world application, we propose that the contemporary pricing policies for banking and mobile telephone services are optimal under boundedly rational behaviour based on positional ambiguity.

Only when the degree of pessimism and the level of loyalty are sufficiently high, the middleman will act in a competitive fashion and conform to setting low extraction rates. Hence, only when platform competition is the norm, competitive behaviour emerges. This is usually the case in, for example, the retail industry.

We explore these conclusions further in an activity level formulation of our model, which allows us the make inferences about hybrid cases of a high degree of pessimism under low levels of loyalty as well as a low degree of pessimism under high levels of loyalty.

\subsection{Relationship to the literature}

\noindent
In economics, the importance of middlemen in intermediated interaction situations has already been recognised in \citet{KalaiMiddlemen1978}, investigating payoffs to middlemen in core allocations in an intermediated trade situation. \citet{RubinsteinWolinsky1987} subsequently extended this analysis in a model of market search. \citet{Biglaiser1993} and \citet{BiglaiserFriedman1994} seminally introduced middlemen in an industrial organisations approach. 

\citet{JacksonWolinsky1996} seminally analysed middlemen within the context of a network model of economic interaction, which was extended by \citet{GillesChakrabarti2006}. This literature showed that middleman positions are critical in networked intermediation and that such middlemen can extract significant gains from their positions within a cooperative game theoretic framework. Recently \citet{Masters:2007tu}, \citet{Masters:2008ud}, \citet{Watanabe:2010uv} and \citet{Watanabe:2011} expanded this approach to develop a network approach to search and matching in markets.

In the literature on two-sided markets, the notion of a \emph{platform} was seminally introduced in \citet{Tirole2003}. The developments in this strand of literature focussed mainly on the pricing of intermediation services and anti-trust policy analysis. For a survey for two-sided platforms we refer to \citet{Tirole2006}, \citet{EvansSchmalensee2007}, \citet{Rysman2009} and \citet{Damme2012} and for multi-sided platforms to \citet{HagiuWright2011} and \citet{EvansSchmalensee2013}.

The platform literature differs in its approach fundamentally from our investigation in this paper. We are neither concerned with platform pricing nor with anti-trust implications of platform provision. Instead, we focus on the monopolistic extraction of rents from the intermediating position of a platform provider, in network terminology known as a ``middleman''. In our model, the pricing decision is deliberately simplistic and the analysis is purely structural: The middleman acts in response to perceived positional features of the network.

\section{Preliminaries}

\noindent
In this section we develop some necessary tools for building a game theoretic model of middleman contestability for general intermediated interaction.

\paragraph{Strategic form games}

Let $N = \{ 1, \ldots ,n \}$ be interpreted as a finite set of \emph{players\/} representing rational decision makers. Each player $i \in N$ is assigned a \emph{strategy set\/} $S_i$ with typical strategy $x_i \in S_i$. An ordered list $(x_1, \ldots ,x_n) \in S_1 \times \cdots \times S_n \equiv S$ is called a \emph{strategy profile,\/} thus introducing $S$ as the set of strategy profiles. Finally, for each player $i \in N$ we introduce a \emph{payoff function\/} $\pi_{i} \colon S \to \mathbb{R}$. A \emph{game} can now be represented as a pair $(S, \pi)$, where $\pi = (\pi_{1}, \ldots , \pi_{n}) \colon S \to \mathbb{R}^N$ is the payoff function. 

A strategy profile $x^{\ast} = (x^{\ast}_1, \ldots ,x^{\ast}_n) \in S$ is a \emph{Nash equilibrium} of the game $(S, \pi )$ if it holds that for every player $i \in N$ and every strategy $y_i \in S_i \colon$
\begin{equation} \label{eq:NE}
\pi_i (x^{\ast} ) \geqslant \pi_i (x^{\ast}_1 , \ldots ,x^{\ast}_{i-1} , y_i , x^{\ast}_{i+1}, \ldots ,x^{\ast}_n ) .
\end{equation}

\paragraph{Introducing ambiguity into strategic games}

Considering decision-making processes subject to ambiguity regarding a player's game theoretic position, it is natural to take into account how each of the players perceives the actions of each other. This refers to the ambiguity in each player's mind about how other players might come to decisions. This form of positional ambiguity can be captured by an appropriately constructed equilibrium notion, an \emph{ambiguity equilibrium\/} concept based on so-called ``neo-additive belief systems''.

Based on the seminal contributions in \citet{EichbergerKelsey2000}, \citet{EichbergerKelseySchipper2007} show that under neo-additive beliefs a pure strategy ambiguity equilibrium in a two-player game corresponds to a pure strategy Nash equilibrium in a transformed game with a three-part payoff function.\footnote{We point out that \citet{ChateauneufEichbergerGrant2007} provide an axiomatic foundation for the equilibrium concept that underlies the ambiguity equilibrium that was introduced in these papers.} This insight incorporates the consistency of the mutual (ambiguous) strategic expectations of both players about each other's strategy at equilibrium. We use this formulation here to describe the behaviour of the middleman in the intermediated interaction situation.

Formally, consider any non-cooperative game $(N,S, \pi )$. For player $i \in N$ in this game, ambiguity is represented by the quadruplet $(M_{i} , \lambda_{i} ; m_{i} , \gamma_{i} )$ consisting of the following elements.
\begin{description} 
\item[Optimistic beliefs.] Each player $i$ formulates well-defined optimistic expectations with regard to her payoffs in the game. These expectations describe the best that can occur in the game for this player. In case of positional considerations, this refers to the case that the player has a ``high extraction'' position that allows for maximal extraction of gains from the social interaction situation.
\\
The \emph{optimistic payoff function\/} of player $i$ is a function $M_{i} \colon S_{i} \to \mathbb{R}$ assigning to every strategy $x_{i} \in S_{i}$ of player $i$ a high extraction payoff $M_{i} (x_{i}) \in \mathbb{R}$.
\\
The number $\lambda_{i} \in [0,1]$ represents the weight that player $i$ puts on her beliefs that high extraction will occur. In other words, $\Lambda_i$ represents the \emph{degree of optimism\/} of player $i$. If $\lambda_{i}=0$, player $i$ has no expectation that she will receive maximal payoffs in the game, while $\lambda_{i}=1$ refers to the other extreme case that player $i$ is fully convinced that she will only receive maximal payoffs.

\item[Pessimistic beliefs.] Similarly, each player $i$ formulates pessimistic expectations with regard to her payoffs in the game in a dual fashion to her optimistic beliefs. These expectations describe the worst case that this player imagines facing. In case of positional considerations, this refers to the case that the player has a ``minimal extraction'' position that allows for only minimal payoffs from social interaction.
\\
The \emph{pessimistic payoff function\/} of player $i$ is represented by a function $m_{i} \colon S_{i} \to \mathbb{R}$ assigning to every strategy $x_{i} \in S_{i}$ of player $i$ her minimally expected payoff $m_i (x_{i}) \in \mathbb{R}$.
\\
The number $\gamma_{i} \in [0,1]$ represents the weight that player $i$ puts on her pessimistic beliefs, i.e., player $i$'s \emph{degree of pessimism}. If $\gamma_{i}=0$, player $i$ assigns no weight to the possibility that she will receive minimal payoffs in the game, while $\gamma_{i}=1$ refers to the other extreme case that player $i$ is fully convinced that she will only receive minimal payoffs. 
\end{description} 
Throughout we assume that the belief-system $(\lambda_{i} , M_{i} ; \gamma_{i} , m_{i})_{i \in N}$ is \emph{proper\/} in the sense that for every player $i$ it holds that $\lambda_{i} + \gamma_{i} \leqslant 1$, where the sum of the degrees $\lambda_{i} + \gamma_{i}$ is interpreted as the \emph{degree of ambiguity\/} of player $i$. The remainder $1 - \lambda_i - \gamma_i \geqslant 0$ represents the player's degree of belief that she is in a ``regular'' state in this strategic interaction situation.
\begin{definition} 
A strategy profile $x^\star \in S$ is an \textbf{ambiguity equilibrium} in the  game $(S, \pi )$ for the proper belief system $(\lambda_{i} , M_{i} ;  \gamma_{i} , m_{i})_{i \in N}$ if $x^\star$ is a Nash equilibrium in the modified game $(S, \overline{\pi})$, where $\overline{\pi}_{i} \colon S \to \mathbb{R}$ for each player $i$ is a modified payoff function given by
\begin{equation}\label{eq:Choquet-def} 
\overline{\pi}_i (x_i, x_{-i}) = \lambda_{i}M_{i}(x_{i}) + \gamma_{i}m_{i}(x_{i}) + (1-\gamma_{i}-\lambda_{i}) \pi_{i}(x_{i},x_{-i}). 
\end{equation} 
\end{definition} 
If the degree of ambiguity is zero, the modified payoff formulation (\ref{eq:Choquet-def}) reduces to the standard  payoff function.

\section{A hedonic model of intermediated interaction}

\noindent
In this section we develop a general hedonic model of an intermediated interaction situation. We assume there are two users and a single middleman who provides a platform to facilitate the interaction between these two users as depicted in Figure 1. Both users select a certain interaction level in the provided platform. The middleman in turn sets a participation fee for each of these users. This participation fee is set individually, thereby allowing complete price discrimination.

We first introduce a standard game theoretic model in which in equilibrium the middleman will extract maximally from her position as a platform provider. Subsequently, we introduce ambiguity and derive that a mixed situation emerges with full monopolistic extraction or fully competitive behaviour under conditions describing contestation, depicted in Figure 2.

\subsection{A benchmark hedonic approach to platform pricing}

\noindent
We consider a three-player non-cooperative game to describe intermediated interaction. There are two \emph{users,\/} denoted by $1$ and $2$, who aim to interact to generate mutual gains. This interaction is intermediated and both users set a participation level in the interaction platform provided by a single middleman. The generated gains depend on the selected participation levels. Formally, this is described by the strategy sets $S_1=S_2=[0,1]$, with generic strategic variables $s_1$ and $s_2$.\footnote{The assumption that the users' strategies are continuous variables allows us to interpret the users as standing in for two continua of agents that use the interaction platform.}

The middleman $M$ provides an interaction platform on which the two users can generate mutual gains from interactions. The middleman $M$ sets an \emph{access fee\/} for each user, represented by $\rho_i \in [0,1]$, $i\in\{1,2\}$, i.e., $M$'s strategy set is $S_M = [0,1]^2$ with typical element $\rho \equiv (\rho_1 , \rho_2)$.

Regarding the two users, let the function $f_i \colon S_1 \times S_2 \to \mathbb{R}_{+}$ describe the gross benefits generated for user $i \in \{ 1,2 \}$, determined by the selected participation levels $(s_1 , s_2 )$ in the provided interaction platform. Hence, $f_i (s_1,s_2) \geqslant 0$ describes the level of the gains from interaction generated for user $i$ when the users put participation levels $s_1$ and $s_2$ into the provided interaction platform.

We impose the following condition on the benefit functions $f_1$ and $f_2$.
\begin{axiom}\label{axiom-on-f}
For each user $i \in \{ 1,2 \}$, the benefit function $f_i$ is strictly increasing in each argument. 
\end{axiom}
This monotonicity condition represents that gains from interaction are optimally generated at maximal exertion of effort or maximal investment. In other words, we assume that marginal benefits are not exhausted until the assumed maximal level of participation $s_1=s_2=1$ is reached.
\begin{remark}
A special case of our construction is the symmetric situation, where $f_1=f_2=f$ for some strictly increasing benefit function $f \colon [0,1]^{2} \to \mathbb{R}_{+}$. Clearly, Axiom \ref{axiom-on-f} is satisfied in this case.
\\
A canonical example of such a benefit function $f$ is described by the familiar generic Cobb-Douglas formulation $f (s_{1},s_{2}) = s_{1}^{\alpha}\cdot s_{2}^{\beta}$, where $\alpha , \, \beta  >0$. In particular, the case of $\alpha = \beta =1$ generates a symmetric, multiplicative model of the interaction between the two users. We analyse this example later.
\end{remark}
We complete the description of the payoff function of each user by assuming that the participation fee is simply subtracted from the generated benefits for each user. Hence, each user $i \in \{ 1,2 \}$ pays the participation fee $\rho_i \geqslant 0$ straightforwardly from the generated benefit described by $f_i$. For every strategy vector $(s_1,s_2,\rho ) \in S_1 \times S_2 \times S_M$ this introduces the game theoretic payoff function $\pi_i \colon S_M \times S_1 \times S_2 \to \mathbb{R}_+$ of user $i \in \{ 1,2 \}$ as
\begin{equation}
\label{ }
\pi_i (\rho,s_1,s_2) =
\left\{ \begin{array}{ll}
f_i (s_1,s_2) - \rho_i & \mbox{if } \rho_i \leqslant f_i (s_1, s_2 ), \\
0 & \mbox{otherwise}.
\end{array}
\right.
\end{equation}
Next we turn to the development of the game theoretic payoff function of the middleman $M$. We introduce a function $\pi \colon S_M \times [0,1]^2 \to \mathbb{R}_+$ that describes the net income generated for the middleman $M$ from providing an interaction platform at certain set participation fees. We hypothesise that this net income is positive and a function of the participation fees $\rho_1$ and $\rho_2$ as well as the two participation levels $s_1$ and $s_2$.

We make the following regularity assumption.
\begin{axiom}\label{axiom-on-pi}
The net income function $\pi$ is weakly increasing in each argument. 
\end{axiom}
Note that the two users $1$ and $2$ are required both to participate to generate any income for the middleman. Hence, the participation fees as set by the middleman $M$ should not exceed the generated benefits for each of the two users. Therefore, the game-theoretic payoff function for the middleman $M$ is formulated as
\begin{equation}
\pi_M (\rho,s_1,s_2) = 
\left\{ \begin{array}{ll}
\pi (\rho, s_1 , s_2) & \mbox{if }  \rho_1 \leqslant f_1 (s_1,s_2) \mbox{ and } \rho_2 \leqslant f_2 (s_1,s_2)  , \\
0 & \mbox{otherwise}.
\end{array}
\right.
\end{equation}
This completes the setup of the game theoretic framework that describes a hedonic model of intermediated interaction between two users. Next we turn to the analysis of this non-cooperative game in strategic form.
\begin{theorem} \label{thm:max usage}
The following statements hold for our model.
\begin{numm}
\item Under Axiom \ref{axiom-on-f}, $\hat{s}_1= \hat{s}_2 =1$ are weakly dominant strategies for users $1$ and $2$, respectively.

\item Under Axioms \ref{axiom-on-f} and \ref{axiom-on-pi}, the strategy tuple of maximal usage and full extraction described by $\hat{s}_1= \hat{s}_2  =1$ and $(\hat{\rho}_1,\hat{\rho}_2) = (f_1(1,1),f_2(1,1))$ is a Nash equilibrium in the game describing the hedonic intermediated interaction situation.

\item Under Axioms \ref{axiom-on-f} and \ref{axiom-on-pi}, the strategy tuple of maximal usage and full extraction described by $\hat{s}_1= \hat{s}_2  =1$ and $(\hat{\rho}_1,\hat{\rho}_2) = (f_1(1,1),f_2(1,1))$ is Pareto efficient.
\end{numm}
\end{theorem}

\begin{proof}
From Axiom \ref{axiom-on-f} it follows immediately that $\max \, f_i (\cdot , s_{-i} ) - \rho_i$ is maximised by $s_i =1$ for both $i=1,2$ regardless of $s_{-i} \in [0,1]$ and $\rho_i \geqslant 0$. Furthermore, if $\rho_i > f_i (1, s_{-i} )$, then $\pi_i =0$ irrespective of $s_{-i} \in [0,1]$. Hence, indeed $\pi_i$ is maximal at $s_i=1$ irrespective of $s_{-i} \in [0,1]$ and $\rho_i \geqslant 0$. This shows assertion (i).
\\
From (i) it follows that $\hat{s}_1= \hat{s}_2 =1$ are best responses to any $\rho = (\rho_1 , \rho_2 ) \geqslant 0$. Given the formulation of $\pi_M$ under Axiom \ref{axiom-on-pi}, it is obvious that $\rho_i = \hat{\rho}_i = f_i (1,1)$, $i\in\{1,2\}$ is a best response to $\hat{s}_1= \hat{s}_2 =1$, showing assertion (ii).
\\
Obviously, by Axiom \ref{axiom-on-pi}, $\pi_M$ is globally maximal over $(s_1,s_2,\rho)$ at $\rho_i = \hat{\rho}_i = f_i (1,1)$, $i\in\{1,2\}$, if $s_1=s_2=1$. The strict monotonicity of $f_i$ in both $s_1$ and $s_2$ implies in turn that $\pi_i$ is maximal over $(s_1,s_2)$ for the given $\rho_i = \hat{\rho}_i = f_i (1,1)$, $i\in\{1,2\}$. This shows assertion (iii).
\end{proof}

\medskip\noindent
Theorem \ref{thm:max usage} shows that the maximal usage of the interaction platform described by $\hat{s}_1= \hat{s}_2  =1$ and its maximal exploitation by the middleman represented by access fees $\hat{\rho}_i = f_i (1,1)$, $i\in\{1,2\}$, is individually stable as a Nash equilibrium as well as Pareto efficient under the relatively weak hypotheses stated in Axioms \ref{axiom-on-f} and \ref{axiom-on-pi}. This state of maximal usage and full exploitation represents the benchmark case, in which the middleman exercises maximal monopolistic control of the intermediated interaction taking place in the provided platform. As such, it points to the expected outcome of intermediated interaction described in our model in the absence of meaningful contestation.

\subsection{Introducing contestation as ambiguity}

\noindent
Next we introduce ambiguity into our simple hedonic model of intermediated interaction to express the state of contestation of the middleman's position in this interaction. Hence, we endow the middleman with optimistic and pessimistic beliefs and degrees of optimism and pessimism, according to the neo-additive formulation of \citet{EichbergerKelseySchipper2007}.

Thus, the only player considered to be affected by positional ambiguity is the middleman, who is ambiguous about her position in the intermediated interaction situation. Her optimistic beliefs refer to a fully successful middleman in complete control, who is able to extract maximal gains from her position. This state obviously corresponds to the Nash equilibrium identified in Theorem \ref{thm:max usage}(ii).

Her pessimistic beliefs refer to a contested middleman position in which the users have access to at least one alternative interaction platform. We model this through the emergence of reduced expected interaction levels in the middleman's platform and a lesser ability to extract full participation fees from the provision of that platform.

We introduce these two sets of beliefs as functions $M$ and $m$ in accordance to the ambiguity model of \citet{EichbergerKelseySchipper2007} summarised in Section 2.2.
\begin{description}
\item[Optimistic beliefs:] The middleman believes that the two users fully participate in his interaction platform at the maximal level $\hat{s}_1=\hat{s}_2=1$. Hence, the optimistic outcome for the middleman $M$ equals
\[
M (\rho ) = \left\{ \begin{array}{ll}
\pi (\rho,1,1) & \mbox{if } \rho_1 \leqslant f_1(1,1) \mbox{ and } \rho_2 \leqslant f_2(1,1), \\
0 & \mbox{otherwise}
\end{array}
\right.
\]
\item[Pessimistic beliefs:] Under full contestation, the middleman assumes that the two users only partially participate in the provided platform. This remaining activity can be interpreted as a \emph{loyalty interaction level}, represented by $s_1 = s^*_1$ and $s_2 = s^*_2$. The situation that there is no loyalty at all can be described by $s^*_1 = s^*_2 =0$.
\\
These levels express that each user might reduce his participation in the platform and use alternatively provided platforms as well. We assume that the middleman has an assessment of the contested situation in that she assumes that usage might be reduced, but not necessarily all the way to zero. Hence, the parameters $s^*_1$ and $s^*_2$ represented such loyalty interaction levels are part of the description of the (pessimistic) beliefs of the middleman.
\\
Now for the given usage levels, the generated pessimistic payoff function for the middleman $M$ can be formulated as
\[
m( \rho ) = 
\left\{ \begin{array}{ll}
\pi (\rho, s^*_1, s^*_2) & \mbox{if } \rho_1 \leqslant f_1(s^*_1,s^*_2), \mbox{ and } \rho_2 \leqslant f_2(s^*_1,s^*_2) \\
0 & \mbox{otherwise}.
\end{array}
\right.
\]
\end{description}
Given the introduced beliefs and a degree of optimism $\lambda$ and a degree of pessimism $\gamma$, we now arrive at the following modified payoff function of the middleman $M$:
\[
\overline{\pi}_M (\rho,s_1,s_2) = \lambda M(\rho) + \gamma m(\rho) + (1- \lambda - \gamma ) \pi_M (\rho,s_1,s_2).
\]
Our main result shows that if the middleman is ambiguous, she might maintain maximal access fees even in the case that she might expect to be contested and that there is little loyalty among her customers.
\begin{theorem} \label{thm:contested case}
Assume that Axioms \ref{axiom-on-f} and \ref{axiom-on-pi} hold, that $\gamma < 1$, and that $s^*_1 <1$ as well as $s^*_2 <1$. Then, in the resulting ambiguity equilibrium, the middleman charges the full exploitation fees of $(\rho_1,\rho_2) = (f_1(1,1),f_2(1,1))$ if and only if
\begin{equation}
\label{eq:contested case}
\Delta \geqslant \frac{\gamma}{1- \gamma} \, \pi (\varphi , s^*_1,s^*_2) 
\end{equation}
where $\varphi = (f_1(s^*_1,s^*_2),f_2(s^*_1,s^*_2))$ represents the full extraction at the pessimistic, low usage level of the interaction platform and
\[
\Delta = \pi (f_1(1,1),f_2(1,1),1,1) - \pi (\varphi ,1,1)
\]
is the middleman's maximal income differential between the maximal usage and low usage levels.
\end{theorem}

\begin{proof}
Define for ease of notation $F = (f_1(1,1),f_2(1,1)) \in S_M$ as the maximal fee vector and $\varphi = (f_1(s^*_1,s^*_2),f_2(s^*_1,s^*_2)) \in S_M$ as the maximal fee vector in case of an interaction level at the assumed loyalty levels $s^*_1$ and $s^*_2$.
\\
First note that, from Axiom \ref{axiom-on-f}, $f_1$ and $f_2$ are strictly increasing. Hence, by these hypotheses $\varphi \ll F$.
\\
Since both users are not subjected to ambiguity, Theorem \ref{thm:max usage}(i) holds and, therefore, in the resulting ambiguity equilibrium both users will still access the provided platform maximally at $s_1=s_2=1$. We can now compute that the ambiguity payoff for the middleman selecting the maximal fee $\rho = F$ would result into

\begin{align*}
\overline{\pi}_M (F,1,1) & = \lambda \cdot \pi (F,1,1) + \gamma \cdot 0 + (1- \lambda - \gamma ) \pi (F,1,1) = \\
& = (1- \gamma ) \pi (F,1,1) .
\end{align*}

Similarly, for the reduced fee $\rho = \varphi$ we compute

\begin{align*}
\overline{\pi}_M (\varphi ,1,1) & = \lambda \cdot \pi (\varphi ,1,1) + \gamma \cdot \pi (\varphi ,s^*_1,s^*_2) + (1- \lambda - \gamma ) \pi (\varphi ,1,1) = \\
& = \gamma \cdot \pi (\varphi ,s^*_1,s^*_2) + (1- \gamma ) \pi (\varphi ,1,1) .
\end{align*}

Hence, $\overline{\pi}_M (F,1,1) \geqslant \overline{\pi}_M (\varphi ,1,1)$ if and only if
\[
(1- \gamma ) \pi (F,1,1) \geqslant \gamma \cdot \pi (\varphi ,s^*_1,s^*_2) + (1- \gamma ) \pi (\varphi ,1,1)
\]
if and only if
\[
(1- \gamma ) \left( \, \pi (F,1,1) - \pi (\varphi ,1,1) \, \right) \geqslant \gamma \cdot \pi (\varphi ,s^*_1,s^*_2) .
\]
This is equivalent to
\[
\Delta = \pi (F,1,1) - \pi (\varphi ,1,1) \geqslant \frac{\gamma}{1-\gamma} \, \pi (\varphi ,s^*_1,s^*_2) .
\]
This completes the proof of the assertion.
\end{proof}

\section{Application: An activity level formulation}

\noindent
In this section we develop a simplification of the generic model. This application is founded on the premise that the middleman uses an activity level measure to gage the success of her business. Furthermore, we assume that in principle the middleman considers the success of her platform as essential in her performance and, therefore, in her game theoretic payoff function $\pi$.

Let $g\colon [0,1]^{2} \to \mathbb{R}_{+}$ capture the total \emph{activity level} that the middleman perceives as being created by the users' participation in the provided interaction platform. Hence, $g (s_1,s_2) \geqslant 0$ measures the activity level perceived by the middleman $M$ if users $1$ and $2$ exert participation levels $s_1$ and $s_2$, respectively. The next hypothesis is equivalent to Axiom \ref{axiom-on-pi}.
\begin{axiom}\label{axiom-on-g}
The activity measurement function $g$ is weakly increasing in each argument. 
\end{axiom}
The simplest example of a weakly monotonic activity measurement function is the constant function given by $g(s_1,s_2)=1$ for all $s_1,s_2 \geqslant 0$. This particular case refers to the absence of the measurement by the middleman.

We assume now that the middleman $M$ measures her success by the income generated as well as the perceived level of activity $g (s_1,s_2)$ in the provided interaction platform. We use a simple multiplicative formulation.
\[
\pi_M (\rho,s_1,s_2) = 
\left\{ \begin{array}{ll}
(\rho_1 + \rho_2) \cdot g(s_1 , s_2) & \mbox{if } \rho_1 \leqslant f_1 (s_1,s_2) \mbox{ and } \rho_2 \leqslant f_2 (s_1,s_2) , \\
0 & \mbox{otherwise}.
\end{array}
\right.
\]
In this formulation it is again taken into account that the platform will only be active if both users are not overcharged, i.e., if $\rho_i \leqslant f_i (s_1,s_2)$ for both $i =1,2$. If at least one of the users is overcharged, the resulting payoffs are assumed to be nil to the middleman $M$.

Furthermore, the middleman's payoff should be interpreted as a utility function; the middleman cares about the usage of the provided platform in proportion to the money the middleman receives, which is $\rho$. We can interpret this feature of the model as capturing network effects and the future profitability of the provided platform.\footnote{In social networks, the usage of the network is a commonly cited statistic to gauge the health of a network.}


We emphasise that under contestation there still results a ``residual'' activity level $g (s^*_1,s^*_2) \geqslant 0$ based on the loyalty of the middleman's customers. This residual activity level represents the minimal activity level that occurs in the provided interaction platform.

Since Axiom \ref{axiom-on-g} and the formulation of $\pi_M$ above implies Axiom \ref{axiom-on-pi}, we can easily derive the following corollary to Theorem \ref{thm:contested case}.
\begin{corollary} \label{coro:act-level}
Assume that Axioms \ref{axiom-on-f} and \ref{axiom-on-g} hold, that $\gamma < 1$, and that $s^*_1 <1$ as well as $s^*_2 <1$. Then in the ambiguity equilibrium for the activity level formulation, the middleman $M$ charges the full exploitation fee $\rho = (f_1(1,1),f_2(1,1))$ if and only if
\begin{equation}
\label{eq:act level}
\left( \, \frac{f_1(1,1) + f_2(1,1)}{f_1(s^*_1,s^*_2) + f_2(s^*_1,s^*_2)} -1 \, \right) \frac{g(1,1)}{g (s^*_1, s^*_2)} \geqslant \frac{\gamma}{1- \gamma} .
\end{equation}
\end{corollary}
\begin{proof}
Note that Axiom \ref{axiom-on-g} implies that $\pi (\rho ,s_1,s_2) = (\rho_1 + \rho_2) \cdot g(s_1,s_2)$ satisfies Axiom \ref{axiom-on-pi}. Therefore, for the activity level formulation Theorem \ref{thm:contested case} applies. Let $\hat{f} \equiv f_1(1,1) + f_2(1,1)$ and $\hat{\varphi} \equiv f_1(s^*_1,s^*_2) + f_2(s^*_1,s^*_2)$. Note that (\ref{eq:contested case}) translates now to
\[
\hat{f} \cdot g(1,1) -\hat{\varphi} \cdot g(1,1) \geqslant \frac{\gamma}{1-\gamma} \, \hat{\varphi} \cdot g(s^*_1,s^*_2) .
\]
This implies the inequality (\ref{eq:act level}) as asserted in the proposition.
\end{proof}

\medskip\noindent
Within the derived formulation in this corollary, the terms on the left measure the effects of being in a contested state using two ratios in terms of the users' perceived usage of the platform and the middleman's measurement of the same. The ratio on the right refers to the degree of pessimism of the middleman. This has two qualitative consequences.

First, the higher the perceived drop in usage of the platform due to competition, the more likely that the inequality holds and, therefore, the middleman will resort to full exploitation of her position.

Second, the higher the degree of pessimism of the middleman, the less likely the inequality holds and, therefore, the less likely the middleman will resort to full exploitation of her position. This indicates that more pessimistic middlemen are more competitive rather than exploitive.

This dual conclusion is further illustrated for a simplification of the general activity level formulation to a benchmark case in which the maximal usage of the platform is normalised to unity and that middleman as well as users measure the usage level under competition by the same quantifier $\sigma \leqslant 1$. The result is a simplified statement with only two relevant variables, the degree of pessimism $\gamma$ and the perceived usage level under competition $\sigma$.
\begin{corollary} \label{coro:contested case}
Assume that Axioms \ref{axiom-on-f} and \ref{axiom-on-g} hold. Furthermore, assume that $f_1(1,1) = f_2(1,1) = g(1,1) = 1$ as well as that there is a residual activity level given by $g(s^*_1,s^*_2) = f_1(s^*_1,s^*_2) = f_2(s^*_1,s^*_2) \equiv \sigma$. Then in the resulting ambiguity equilibrium the middleman charges the full exploitation access fee $\rho = (1,1)$ if and only if
\begin{equation}
(1- \gamma) (1 - \sigma ) \geqslant \gamma \sigma^2 .
\end{equation}
\end{corollary}
This corollary is illustrated in the graph in Figure \ref{fig:2D}. The shaded area is exactly the set of those values of $\sigma$ and $\gamma$ in which the middleman will charge the full exploitation fee of $\rho=(1,1)$.

\begin{figure}
\centering
\includegraphics[scale=1.2]{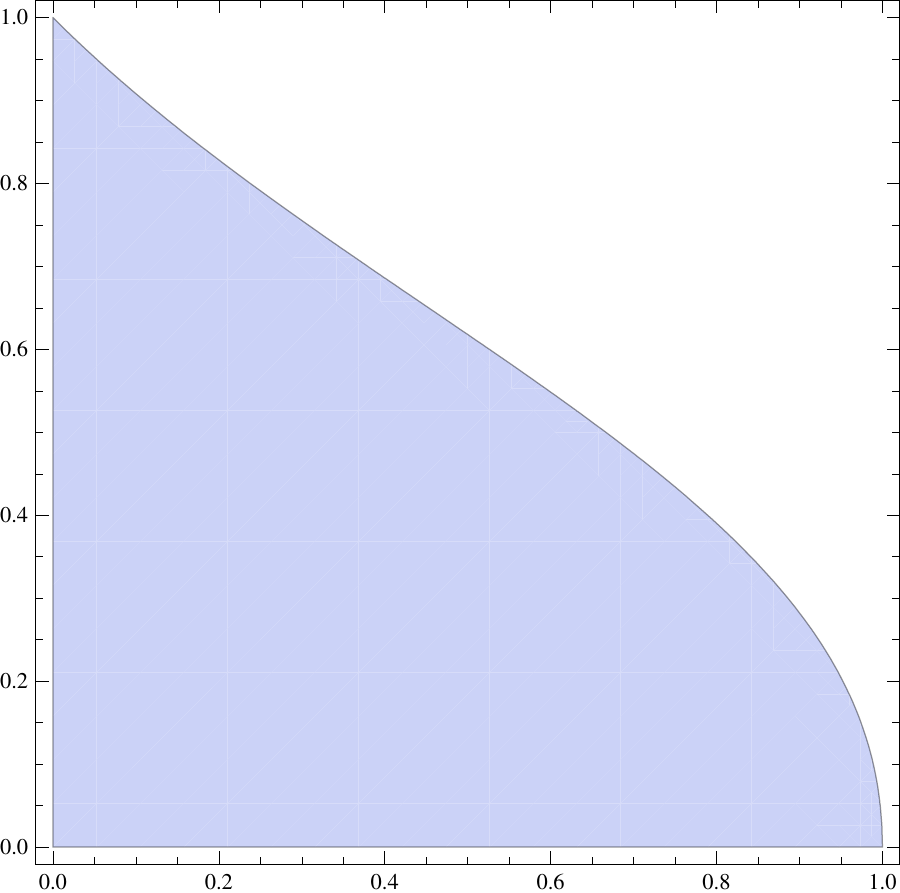}
\caption{Depiction of $\gamma$ (X-axis) and $\sigma$ (Y-axis) satisfying Corollary \ref{coro:contested case}.}
\label{fig:2D}
\end{figure}

This graphical illustration of Corollary \ref{coro:contested case} shows the two qualitative statements made above. Higher usage losses are represented by lower levels of residual activity $\sigma$ and result in full exploitation. Similarly, lower degrees of pessimism $\gamma$ result in full exploitation as well.

These two qualitative effects seem at first counterintuitive. The shaded area predominates, and in this area the middleman acts as though oblivious to contestation. Consider a low value of residual activity $\sigma$---on the vertical axis. The first reaction, in the standard microeconomic way of thinking, is to say that a low level of residual activity $\sigma$ will motivate the middleman to match it; that would be the middleman's way of competing with the perceived threat of undercutting that the low $\sigma$ represents.

We offer the following explanation for this seemingly paradoxical result. The middleman who faces a low residual activity level $\sigma$ chooses to ignore the contested situation this $\sigma$ represents unless her degree of pessimism $\gamma$ is sufficiently large. This behaviour is in fact understandable if we think in terms of the middleman grasping at straws. ``I feel my middleman position is contested, which means my price-setting power may be weak; well then, I will try to grab as much positional rent as I can now while I have some monopolistic power.'' This explanation is in keeping with the bounded rationality nature of the middleman's process of decision-making under ambiguity. While fleshing out this idea fully requires a dynamic model and indeed one with at least two distinct middlemen, which we plan to present in a subsequent paper, we feel that the idea captures powerfully behaviours we see around us.

For example, consider the competitive practices of \emph{Microsoft}. It continues to rely for the vast majority of its revenue on its Windows and Office divisions---this has been notable for some time now---see, for example, \citet{Economist2006}. Meanwhile the patterns of use of computing devices are moving fast away from the Windows platform running on desktop and laptop computers, towards mobile computing devices on which Microsoft has so far failed to establish a strong presence, thereby being unable to act as a middleman (platform provider) to the users and programmers interested in mobile computing. Microsoft has been acting to protect its privileged role as a platform provider for some time now, notably resulting in the Antitrust case brought against it in the US. (See \citet{Cass2013} for a discussion in the general context of antitrust in high-tech.) The recent troubles of \emph{Research In Motion} (RIM), the maker of the Blackberry and the associated mobile communication platform, appears to fit our analysis and can be seen as another example in the field of information technology.

On the other hand, the retail industry can be represented on the other end of the two dimensional spectrum considered here. Indeed, retail firms such as supermarket chains can reasonably be expected to be in the unshaded part of the graph, with relatively high $\sigma$ and relatively high $\gamma$. The high degree of pessimism $\gamma$ comes from the normality of strong price competition in retail, implying that contested positions are expected as ``normal''. Furthermore, the high expected usage level $\sigma$ comes from the high degree of trust that well-established chains command with their customers, who continue to use them in fairly high degree even when alternatives are present.

\section{Some concluding remarks}

\paragraph{Considering user externalities}

We can explicitly introduce externalities among the users of a platform through a simple modification of our model. This refers to the simple observation that only if both users participate, there actually will be any generation of mutual benefits to these users and, by extension, to the platform provider. In our approach thus far this is excluded through imposition of Axioms \ref{axiom-on-f} and \ref{axiom-on-pi}.

A simple alternative formulation to illustrate this point is that of a straightforward Cobb-Douglas activity level function in the setting of the activity level model introduced in Section 4. This allows for a fully symmetric formulation of the fact that benefits are only generated if both users actually participate, i.e., if $s_1, \, s_2 >0$.

More formally, within the activity level formulation, we assume the canonical formulation with
\[
f_1 (s_1,s_2) = f_2 (s_1,s_2) = f(s_1,s_2) = g(s_1,s_2) =s_1 \cdot s_2 .
\]
We remark that this case does not satisfy Axiom \ref{axiom-on-f}, since as is the case with the Cobb-Douglas function, we have here a minimal level set at the value zero in which strict monotonicity fails. This describes the inter-personal externalities required for mutually beneficial interaction between the two users.

The state of maximal usage and exploitation is again given by the Pareto optimal Nash equilibrium with $\hat{s}_1 = \hat{s}_2 = \hat{\rho}_1 = \hat{\rho}_2 =1$, resulting in payoffs $\hat{\pi}_1 = \hat{\pi}_2 =0$ and $\hat{\pi}_M=1$.

There is also a class of trivial Nash equilibria in which $s_1 = s_2 = 0$ and, given that $f(0,0)=0$, $\rho$ taking any nonnegative value.\footnote{The reason for the existence of these trivial equilibria is that the multiplicative formulation is not strictly monotone on the boundary---as stated above.} In each of these equilibria, no player can improve her payoff by deviating from her equilibrium strategy.

We conclude that the introduction of simple externalities to the interaction between platform users might generate multiplication of equilibria and, consequently, potential coordination failure.


\paragraph{A comparison with fully rational behavioural approaches}

In this paper we have used quite deliberately an approach founded on the theory of ambiguity in decision making. Our main hypothesis is throughout that the middleman's behaviour can be understood from a purely subjective perspective, informed by the middleman's ambiguity about her position in the network.

Our approach contrasts with more traditional models of decision making under incomplete information in which decision makers are not ambiguous in their assessment of the position of the middleman in the intermediated interaction situation. Aumann's model of incomplete information \citep{Aumann1976} and Harsanyi's theory of Bayesian games \citep{Harsanyi1967,Harsanyi1968a,Harsanyi1968b} provide alternative approaches founded on a fully rational processing of informational deficiencies rather than the boundedly rational perspective on ambiguity applied here.

We use a boundedly rational perspective to reflect the principle that competition is actually a `state of mind'. This perspective assumes explicitly that this state of mind is not fully rational, but instead a reflection of a state of fundamental ambiguity about the position of the middleman in the intermediated interaction situation. This also implies that we do not model the potential competitors of that middleman explicitly, but instead limit ourselves to just modelling the state of mind of the middleman only.

Our approach yields surprising insights. We find that full extraction occurs in cases of low ambiguity, reflecting situations in which middleman contestation is not normal. This refers to situations with significant positional power of platform providers such as occur in markets for networked goods. Instead, competitive pricing occurs in situations that such contestation is normal, i.e., the degree of ambiguity is high.  The latter refers to the practices seen in the retail industry and other competitive market situations.

It is this explanatory power of our approach that justifies it.

\paragraph{Future developments}

Future research should focus on enhancing our framework to incorporate more elaborate models with multiple platform providers that operate under ambiguity about their position in the prevailing interaction networks.

\singlespace
\bibliographystyle{econometrica}
\bibliography{BibDB}
\end{document}